\documentclass[11pt]{article}
\usepackage[english]{babel}
\usepackage{graphicx}
\usepackage{caption}
\usepackage{subcaption}
\usepackage{scrextend}
\usepackage{thmtools}
\usepackage[normalem]{ulem}

\newcommand{\minp}{\rho}
\newcommand{\len}{\ell}
\newcommand{\sizefrac}{\beta}

\newcommand{\minpexp}{3000}
\newcommand{\delexp}{3100}
\newcommand{\alphaexp}{4/3}
\newcommand{\lenexp}{30}
\newcommand{\phiexp}{10}

\newcommand{\minpval}{d^{-60}\eps^{\minpexp}}
\newcommand{\lenval}{d^6 \eps^{-\lenexp}}
\newcommand{\sizefracval}{\eps/10}
\newcommand{\deltaval}{d^{-70}\eps^{\delexp}}
\newcommand{\alphaval}{\frac{\eps^{\alphaexp}}{300,000}}
\newcommand{\phival}{\eps^{\phiexp}}

\newcommand{\trwalk}{\widehat{M}}

\newcommand{\ball}{\mathrm{IB}}
\newcommand{\cluster}{\texttt{cluster}}

\newcommand{\level}[3]{L_{#1, #2, #3}}
\newcommand{\free}[1]{F_{#1}}

\newcommand{\bT}{{\bf T}}

\newcommand{\ST}{K}
\newcommand{\sizethresh}{k}
\newcommand{\ourmodel}{\texttt{General\-{-}Label\-{-}Model}\xspace}
\newcommand{\oursetting}{\texttt{General\-{-}Label\-{-}Setting}\xspace}

\newcommand{\findr}{{\tt findr}\xspace}
\newcommand{\findanchor}{{\tt findAnchor}\xspace}
\newcommand{\findpart}{{\tt find}\-{\tt Partition}\xspace}

\newcommand{\globpart}{{\tt global}\-{\tt Part}\-{\tt it}\-{\tt ion}\xspace}

\newcommand{\clustercode}{
\noindent
\fbox{\begin{minipage}{0.86\textwidth}
    {\cluster$(v, t, k)$}
    \begin{compactenum}
    \item Determine $\trwalk^{t}\vec{v}$
        \item For all $k' \in [k, 2k]$ calculate $\Phi(\level{v}{t}{k'})$.
        \item Find the largest $k' \in [k,2k]$ (if any) with the following properties: $\Phi(\level{v}{t}{k'} \cup \{v\}) \leq \phi$ and $\level{v}{t}{k'} \in \supp(\trwalk^{t}\vec{v})$.
        \item If such a $k'$ exists, set $C := \level{v}{t}{k'} \cup \{v\}$, else $C := \{v\}$.
        \item Return $C$.
    \end{compactenum}
\end{minipage}
}
}

\newcommand{\globpartcode}{
\begin{center}
\fbox{%
  \begin{minipage}{0.85\textwidth}
  \textbf{\globpart$(\bS_1, \bS_2)$}
  
  \medskip
  \textbf{Preprocessing:}
  \begin{compactenum}
    \item For each phase $h = 1$ to $\overline{h}$:
    \begin{compactenum}
      \item Use $\bS_1(h)$ to generate:
      \begin{compactenum}
        \item A $b$-wise independent collection $\{H_v^{(h)}\}_{v \in [N]}$, where $H_v^{(h)} \sim \texttt{Ber}(\delta)$
        \item A $b$-wise independent collection $\{T_v^{(h)}\}_{v \in [N]}$, where $T_v^{(h)} \sim \texttt{Unif}([1,\ell])$
      \end{compactenum}
    \end{compactenum}
  
    \item For each $v \in V$:
    \begin{compactenum}
      \item Set $h_v := \min\{ h \mid H_v^{(h)} = 1 \}$
      \item Set $t_v := T_v^{(h_v)}$
    \end{compactenum}
  
    \item Call \findr$(\bS_2)$ to obtain thresholds $k_1, \ldots, k_t$
    \item Set $k_v := k_{h_v}$ for all $v$
  \end{compactenum}

  \medskip
  \textbf{Partitioning:}
  \begin{compactenum}
    \item Initialize $\cP \gets \emptyset$, $F \gets V$
    \item For each $v \in V$:
    \begin{compactenum}
      \item Let $C = \cluster(v, t_v, k_v)$
      \item Add connected components of $C \cap F$ to $\cP$
      \item Update $F \gets F \setminus C$
    \end{compactenum}
    \item Output $\cP$
  \end{compactenum}
  \end{minipage}
}
\end{center}
}

\newcommand{\findrcode}{
\noindent
\fbox{
\begin{minipage}{0.85\textwidth}
    {\findr$(\bS_2)$}
    \smallskip
    \begin{compactenum}
    \item For $h = 1$ to $\overline{h}$:
    \begin{compactenum}
        \item Sample $\sizefrac^{-10}$ vertices using $\bS_2$.
        Let $S_h$ be the multiset of sampled vertices
        that are in phase $\geq h$.
        \item If $|S_h| \leq \sizefrac^{-9}/2$, set $k_h = 0$ and continue the for loop. Else, reset $S_h$
        to the multiset of the first $\sizefrac^{-8}$ vertices sampled.
        \item For $k \in [\minp^{-1}]$ and for every $s \in S_h$: \label{step:loop}
        \begin{compactenum}
            \item Compute $C := \cluster(s,t_s,k)$.
            \item For all $u \in C$, use a local procedure from \cite{KSS:21} \emph{(which uses no additional randomness)} to check whether $u \in F_{h-1}$.
            \item If $C$ is not a singleton and $|C \cap F_{h-1}| \geq \sizefrac^3 \sizethresh$, mark
            $s$ as being $(h,\sizethresh)$-viable.
        \end{compactenum}
        \item If there exists some $\sizethresh$ such that there are at least $12\sizefrac^{4}|S_h|$ $(h,k)$-viable vertices, assign an arbitrary
        such $\sizethresh$ as $\sizethresh_h$. Else, assign $\sizethresh_h := 0$. \label{step:setthresh}
    \end{compactenum}
    \item Output $\ST_{\overline{h}} = \{\sizethresh_1, \sizethresh_2, \ldots, \sizethresh_{\overline{h}}\}$.
    \end{compactenum}
\medskip
\textit{Remark:} This step consumes only $\poly(d, 1/\eps)$ time and $O(\sizefrac^{-10} \cdot \log n)$ bits of internal randomness from $\bS_2$. Crucially, it is run \emph{once} and its output is reused across all local simulations.
\end{minipage}
}
}

\newcommand{\parameterinfo}{
\begin{itemize}
    \item[] $\minp = \minpval$: Minimum probability for truncation.
    \item[] $\len = \lenval$: Maximum random walk length.
    \item[] $\sizefrac = \sizefracval$: Unclustered fraction cutoff.
    \item[] $\delta = \deltaval$: Phase probability. 
    \item[] $\alpha = \alphaval$: Heavy bucket parameter.
    \item[] $\phi = \phival$: Conductance parameter.
    \item[] $b = 4 \cdot \ell \rho^{-1}$: Limited randomness parameter.
\end{itemize}
}

\usepackage{graphicx}
\usepackage{amsmath,amssymb,amsthm,mathtools}
\usepackage{paralist}
\usepackage{longtable}
\usepackage{colortbl}
\usepackage{hhline}
\usepackage{bm}
\usepackage{xspace}
\usepackage{url}
\usepackage{fullpage, prettyref}
\usepackage{boxedminipage}
\usepackage{wrapfig}
\usepackage{ifthen}
\usepackage{soul}
\usepackage{color}
\usepackage{xcolor}
\usepackage{framed}
\usepackage{algorithmic,algorithm}
\usepackage[pagebackref,letterpaper=true,colorlinks=true,pdfpagemode=none,urlcolor=blue,linkcolor=blue,citecolor=violet,pdfstartview=FitH]{hyperref}
\usepackage{fullpage}
\usepackage{enumitem}   
\usepackage{ulem}

\usepackage{thmtools}
\usepackage{thm-restate}
\newtheorem{theorem}{Theorem}[section]

\newtheorem{claim}[theorem]{Claim}
\newtheorem{corollary}[theorem]{Corollary}
\newtheorem{problem}[theorem]{Open Problem}
\newtheorem{definition}[theorem]{Definition}
\newtheorem{proposition}[theorem]{Proposition}
\newtheorem{fact}[theorem]{Fact}

\newtheorem{observation}[theorem]{Observation}
\newtheorem{remark}[theorem]{Remark}

\newcommand{\ignore}[1]{}

\DeclareMathOperator*{\supp}{supp}

\newcommand{\cE}{{\cal E}}
\newcommand{\cF}{\mathcal{F}}

\newcommand{\cP}{\mathcal{P}}

\newcommand{\cS}{\mathcal{S}}

\newcommand{\eps}{\varepsilon}

\newcommand{\poly}{\mathrm{poly}}

\newcommand{\bone}{{\bf (1)}}
\newcommand{\btwo}{{\bf (2)}}

\newcommand{\bA}{\boldsymbol{A}}

\newcommand{\bR}{\boldsymbol{R}}
\newcommand{\bS}{\boldsymbol{S}}

\newcommand{\MM}{\mathbb{M}}
\newcommand{\NN}{\mathbb{N}}
\newcommand{\RR}{\mathbb{R}}

\newcommand{\EX}{\hbox{\bf E}}

\newcommand{\eqdef}{:=}

\newcommand{\Sec}[1]{\hyperref[sec:#1]{Sec.\,\ref*{sec:#1}}} %section
\newcommand{\Eqn}[1]{\hyperref[eq:#1]{(\ref*{eq:#1})}} %equation
\newcommand{\Fig}[1]{\hyperref[fig:#1]{Fig.\,\ref*{fig:#1}}} %figure
\newcommand{\Tab}[1]{\hyperref[tab:#1]{Tab.\,\ref*{tab:#1}}} %table
\newcommand{\Thm}[1]{\hyperref[thm:#1]{Theorem\,\ref*{thm:#1}}} %theorem
\newcommand{\Fact}[1]{\hyperref[fact:#1]{Fact\,\ref*{fact:#1}}} %fact
\newcommand{\Lem}[1]{\hyperref[lem:#1]{Lemma\,\ref*{lem:#1}}} %lemma
\newcommand{\Prop}[1]{\hyperref[prop:#1]{Prop.~\ref*{prop:#1}}} %property
\newcommand{\Cor}[1]{\hyperref[cor:#1]{Corollary~\ref*{cor:#1}}} %corollary
\newcommand{\Conj}[1]{\hyperref[conj:#1]{Conjecture~\ref*{conj:#1}}} %conjecture
\newcommand{\Def}[1]{\hyperref[def:#1]{Definition~\ref*{def:#1}}} %definition
\newcommand{\Alg}[1]{\hyperref[alg:#1]{Alg.~\ref*{alg:#1}}} %algorithm
\newcommand{\Ex}[1]{\hyperref[ex:#1]{Ex.~\ref*{ex:#1}}} %example
\newcommand{\Clm}[1]{\hyperref[clm:#1]{Claim~\ref*{clm:#1}}} %example
\newcommand{\Obs}[1]{\hyperref[obs:#1]{Observation~\ref*{obs:#1}}} %example
\newcommand{\Step}[1]{\hyperref[step:#1]{Step~\ref*{step:#1}}} %example
\newcommand{\Rem}[1]{\hyperref[rem:#1]{Remark~\ref*{rem:#1}}} %example
\newcommand{\Assumption}[1]{\hyperref[assm:#1]{Assumption\,\ref*{assm:#1}}} %assumption

\def\commentmode{1}

\ifdefined\commentmode

\else

\fi

\def\internalmacros{1}
\ifdefined\internalmacros

\newcommand{\mypr}{\mathbf{Pr}}
\newcommand{\myex}{\hbox{\bf E}}

\newcommand{\psg}{\mathsf{PRG}}

\renewcommand{\Pr}{\mathbf{Pr}}
\renewcommand{\EX}{\hbox{\bf E}}

\title{Reducing the Randomness in Partition Oracles for Bounded Degree Minor-Free Graphs}

\author{Akash Kumar\thanks{Department of Computer Science, IIT Bombay. { \href{mailto:akash@cse.iitb.ac.in}{akash@cse.iitb.ac.in}}} 
\and Abhiruk Lahiri\thanks{CISPA Helmholtz Center for Information Security, Saarbr\"{u}ken and Heinrich Heine University, D\"{u}sseldorf, Germany. {\href{mailto:lahiri.abhiruk@cispa.de}{lahiri.abhiruk@cispa.de}}}
\and C. Seshadhri\thanks{Department of Computer Science, University of California, Santa Cruz. {\href{mailto:sesh@ucsc.edu}{sesh@ucsc.edu}}.
Supported by NSF DMS-2023495, CCF-1740850, 2402572.} 
}

\date{}
\begin{document}

\maketitle
\thispagestyle{empty}

\abstract{
Consider a bounded-degree graph $G$ that belongs to a minor-closed family (such as planar graphs). Such a graph has a hyperfinite decomposition, wherein, for a sufficiently small $\eps > 0$, one can remove $\eps dn$ edges to obtain connected components of size independent of $n$. (As usual, $n$ is the number of vertices and $d$ is the degree bound.)
In a seminal result, Hassidim-Kelner-Nguyen-Onak (FOCS 2009) introduced the 
\emph{partition oracle}, a procedure that provides local access to a hyperfinite decomposition. 
The partition oracle computes the component containing an input vertex $v$ with query complexity (to $G$) \emph{independent} of $n$. Remarkably, this is done without any preprocessing on $G$.
The coordination is done purely through a shared random seed.

Despite a line of work on optimizing the query complexity of partition oracles, there were no attempts
to bound the size of the random seed. All existing partition oracles use a random seed
of size $\Omega(n)$, which technically implies a linear setup time. Any blackbox derandomization would likely need $\Omega(\log^2n)$ uniform random bits. A natural question is whether the random seed can also have length independent of $n$. 

We prove the $\poly(d\eps^{-1})$-query partition oracles of Kumar-Seshadhri-Stolman can be implemented with a random seed of $\poly(d\eps^{-1}) \cdot \log n$ length.
To get a deeper understanding on the randomness complexity, we consider a more general model where the vertex labels come from the universe $[N]$, where $N \geq n$. In this setting, we prove that any partition oracle even for cycles requires $\omega_N(1)$ random bits.
}

\section{Introduction} \label{sec:introv2}
The study of planar and minor-free graph classes occupies a central place in both structural and algorithmic graph theory. Foundational results by Kuratowski and Wagner characterized planarity through
forbidden minors~\cite{K30, W37}. 
A generalization of these ideas was achieved through the celebrated Graph Minor Theorem of Robertson and Seymour \cite{RS:12, RS:13, RS:20}. On the algorithmic side, 
the linear time planarity algorithm of Hopcroft-Tarjan is a classic result~\cite{HT74}. The classic separator theorems of Lipton-Tarjan and Alon-Seymour-Thomas are standard algorithmic tools for minor-free classes of graphs~\cite{LiptonT:80, AST:94}.
An important implication of these separator theorems is that bounded degree minor-free graphs are \emph{hyperfinite}. For any $\eps > 0$,
one can remove an $\eps$-fraction of the vertices to get connected components of size $O(\eps^{-2})$. 

Minor-freeness and hyperfiniteness have inspired a rich line of work in sublinear algorithms. A seminal 
result of Hassidim-Kelner-Nguyen-Onak introduced the notion of partition oracles for minor-free graph classes~\cite{HKNO}. 
Consider an input graph $G$ with query access to its adjacency list.  A partition oracle is a local, constant query procedure that, given a vertex $v$, returns the connected component of $v$ in a fixed consistent hyperfinite decomposition.
There is no preprocessing of $G$, and all the coordination is done through a fixed random seed. 
At the face of it, it is remarkable that partition oracles even exist;
the proof of hyperfiniteness goes via a global algorithm that recursively applies the planar (or minor-free) separator theorem.
A partition oracle must effectively obtain each piece of this (global) decomposition by only inspecting constant-sized neighborhoods.
Partition oracles are a key primitive in sublinear algorithms on bounded-degree graphs, and can be used
to derive general constant-query testability results for a large class of properties (eg, see Theorem 1.4 in \cite{KSS:21}).
Partition oracles give access to a global object with only local queries, and is thus a classic example
of a Local Computation Algorithm (LCA) \cite{AlonRVX12, RubinfeldTVX11}. 

Despite further work on optimizing the query complexity of partition oracles~\cite{LR15, KSS:21}, all implementations
use a random seed of size $\Omega(n)$. From the perspective of an LCA, the storage for the partition oracle
is (at least) the random seed, and is hence $\Omega(n)$. This somewhat contradicts to sublinear nature of partition oracles, since
linear set up time is required to simply write down this seed. 
Previous work of Alon-Rubinfeld-Vardi-Xie \cite{AlonRVX12} make randomness (and hence space) reduction arguments
for LCAs, and it is likely their tools could be applied. As we discuss later, such arguments that the use
of small-space pseudorandom generators or bit sequences with limited independence can only reduce the random seed 
to size $\Theta(\log^2n)$. There are fundamental reasons why reducing further requires non-blackbox methods.

On the other hand, it is unclear whether a constant number of random bits suffices for partition oracles.
At an extreme, consider a deterministic partition oracle that computes the MD5 hash of its input, and uses
that to simulate a randomized algorithm. It is non-trivial to prove that such an algorithm would fail.
Indeed, maybe augmenting this algorithm with a constant number of random bits (depending only on $\eps$)
would suffice. Given that the query complexity is constant, it seems challenging to prove that $\omega(1)$ bits
of pure randomness are needed.

Our main motivating question is: what is the optimal length of the random seed for partition oracles? 
How much randomness is actually required to implement a partition oracle?

\subsection{Main Result} \label{sec:resultv2}

The standard access model for bounded-degree graphs was introduced by Goldreich-Ron~\cite{GR02}.
Consider an input graph $G = (V,E)$ with $n$ vertices and degree bound $d$. 
Assume that $V = [n]$. The graph is accessed through adjacency list, or neighbor queries:
given $i \leq n$ and $r \leq d$, the oracle outputs the $r$th neighbor of vertex $i$ (or 
returns $\bot$ if it does not exist).  
 
\medskip

We state the definition of a partition oracle, following the notation of \cite{KSS:21}.

\begin{definition} \label{def:oracle} 
Let $\cP$ be a family of graphs with degree bound $d$, and $T: \NN \times (0,1) \to \NN$, $\ell:\NN^2 \times (0,1) \to \NN$ be functions. A procedure $\bA$ is an 
\emph{$(\eps, T(d,\eps), \ell(n,d,\eps))$-partition oracle} for $\cP$ if it satisfies the following properties. The (randomized) procedure takes as 
input random access to $G = (V,E) \in \cP$ (with $n$ vertices), a uniform random seed $r$ of length $\ell(n, d, \eps)$, a proximity parameter $\eps > 0$, and a vertex $v \in V$. (Fixing $G, r, \eps$, we write $\bA_{G,r,\eps}$; all probabilities are over $r$.) The output $\bA_{G,r,\eps}(v)$ is a subset of $V$, and:
\begin{enumerate}
    \item (Consistency) The sets $\{\bA_{G,r,\eps}(v)\}_{v \in V}$ form a partition 
    of $V$, and each set induces a connected subgraph.
    \item (Cut bound) With probability at least $2/3$ over $r$, the number of edges 
    crossing the partition is at most $\eps dn$.
    \item (Query complexity) For every $v$, the procedure $\bA_{G,r,\eps}(v)$ has query complexity $T(d,\eps)$.
\end{enumerate}
\end{definition}

We note that the size of all the sets is at most $T(d,\eps)$, and standard arguments shows that 
the size can be made $O(\eps^{-2})$ \cite{G17-book}. Typically, both $d$ and $\eps$ are treated as constants
(or at least, independent of the size $n$). The seminal work of Hassidim-Kelner-Nguyen-Onak gave a partition oracle for
all minor-closed families with $T(d,\eps) = (d/\eps)^{\poly(d/\eps)}$~\cite{HKNO}. Levi-Ron gave
an improved partition oracle with $T(d,\eps) = (d/\eps)^{\log(d/\eps)}$~\cite{LR15}. Kumar-Seshadhri-Stolman
proved the existence of partition oracles with $\poly(d/\eps)$ query complexity~\cite{KSS:21}. All
of these results have $\ell(n,d,\eps) = \Omega(n)$, and hence require linear randomness.

We prove that the randomness can be reduced to just $O(\log n)$ bits. Note that such a small
seed only suffices to sample a constant number of uniform random vertices. This theorem
is proven by a detailed inspection of the partition oracle of Kumar-Seshadhri-Stolman (henceforth KSS)~\cite{KSS:21}.

\begin{theorem} \label{thm:main}
Let $\cP$ be the set of $d$-bounded degree graphs in a minor-closed family. There exists an $(\eps, \poly(d/\eps), \poly(d/\eps) \cdot \log n)$-partition oracle for $\cP$.
\end{theorem}

We complement this theorem with a lower bound the random seed length in a more general model. Consider the standard Goldreich-Ron setup,
where the labels on $V$ are unknown to the algorithm--the algorithm only knows that all vertex labels come from the set $\{1,2,\ldots, N\}$. There is a \emph{vertex query}
that, given $i \leq n$, fetches the $i$th element in $V$ (according to some arbitrary ordering).
The neighbor query takes as input (the label of) an element in $V$ and a neighbor index $r \leq d$, and outputs
the (label of) the $r$th neighbor of $v$. 

We call this the \emph{\ourmodel}. (We actually prove \Thm{main} in \ourmodel with an $O(\log N)$ bound.) In this model, we can prove a super-constant lower bound for the randomness required for a partition oracle for cycles.

\begin{theorem} \label{thm:main:v3} Consider the \ourmodel with label size being $N$. Let $\eps > 0$ be sufficiently small. Consider any $(\eps, T(d,\eps), \ell(n,d,\eps))$-partition oracle that works for all $n$ cycles where $n \leq N$. 
Then $\ell(n,d,\eps) = \omega_N(1)$.
\end{theorem}

\subsection{Main Ideas and Challenges} \label{sec:ideas}

All existing partition oracles can be thought to ``store" a uniform random $O(\log n)$ length string at every vertex. 
In the case of the earlier partition oracles, this is just a random hash which is used to compare
vertices \cite{HKNO}. For KSS, there is a complex list of vertex-centric parameters that define
``diffusion clusters". When the partition oracle
visits a set of vertices, it uses the stored random strings for its internal randomness. Because
the random string is fixed for every vertex, this ensures consistency when the partition oracle
is called from different vertices. In the discussion below, we treat both $d$ and $\eps$ as constants
and ignore their dependencies.

{\bf Upper bound:} Let us explain how a black box randomness reduction approach would work. Since
the partition oracle has constant query complexity, it only reads $O(\log n)$ bits of randomness,
and only needs $O(\log n)$ bits of storage. Intuitively, as long as these $O(\log n))$ bits are independent,
an individual call to the partition oracle will behave exactly as if the entire source of randomness is fully independent.
This suggests that we can replace entire (uniform) random seed to a sequence of $O(\log n)$-wise independent bits.
Note that this is \emph{not} enough. The cut guarantee of the partition oracles involves global behavior
over all the vertices. With an $O(\log n)$-wise independent bit string, we do not get independent behavior 
for calls to different vertices. Nevertheless, one might hope that the analysis can be carried through.
This is the approach followed in previous work on randomness/space reduction for LCAs \cite{RubinfeldTVX11}.

So we need a $\poly(n)$ length bit string that is $O(\log n)$-wise independent. Standard pseudorandom generators (PRGs)
can be used to generate such sequences, with a uniform random seed of $O(\log^2n)$ length \cite{Nisan90} (Theorem 1). This
bound is also optimal, so we cannot hope for more reduction. An alternate, but morally equivalent,
perspective is to use Nisan's small space PRG \cite{Nisan90}. The partition oracle is a small space algorithm,
and we would need to construct a small space algorithm that verifies the partition oracle guarantees. Again,
we require $\Omega(n)$ bits in total for an $O(\log n)$-space algorithm; which leads to a seed length of $O(\log^2n)$.

To get a smaller seed length requires a different approach. If we wanted an $O(\log n)$-length (uniform) seed,
we can only afford $O(1)$-wise independence. That seems extremely limiting. The probability of any ``bad" event
can only be upper bounded by a constant, so we cannot union bound over the $n$ possible inputs (vertices) to the 
partition oracle. With $O(\log n)$-wise independence, Chernoff bounds would lead to $1/\poly(n)$ error bounds,
which can be used with union bounds to get the global guarantees of partition oracles. With $O(1)$-wise independence,
there is no hope of any union bound.

Our discovery is that, by a careful and detailed inspection, the KSS analysis can be made to go through with $O(1)$-wise
independence. We note that this is somewhat remarkable, since we cannot afford any union bounds over all the vertices.
We believe that the main contribution here is the actual discovery, rather than the analysis. Once we get some basic observations
done, it is mostly a tedious (but somewhat routine) calculation of various steps of the KSS proof.

{\bf Lower bound:} It is surprisingly tricky to prove lower bounds for the random seed length. 
As mentioned earlier, consider the following deterministic construction for these strings. Simply
compute the MD5 (or some cryptographic) hash of the vertex label. These hashes can be used
by an randomized partition oracle for its internal randomness. In general, we would expect
this algorithm to work for most inputs. Indeed, the only way to ``defeat" this algorithm would
be a carefully constructed instance that exploits dependencies among the MD5 hashes. But such an instance seems
nearly impossible to construct. 
This argument also shows that we cannot just compute a single hard family or distribution, and any lower
bound has to exploit the structure of a specific deterministic algorithm. 

The standard tool for lower bounds for randomized algorithms is Yao's minimax lemma. Yet it does
not seem to help to actually bound the length of the random seed. One of the conceptual struggles
the authors faced was in trying to construct some framework that leads to lower bounds for random seed
length.

We take inspiration from a somewhat unlikely source: query lower bounds for monotonicity testing.
Fischer \cite{Fischer04} gave a beautiful application of Ramsey theory to show that adaptivity does
not help for monotonicity testing on the line. Chakrabarty-Seshadhri extended that argument
to other domains \cite{ChakrabartyS13}. The essence of the Ramsey theory argument is that for any algorithm,
there exists a restriction of possible inputs where the algorithm behaves like it is ``comparison-based".
Note that this restriction is fundamentally based on the specific algorithm.

While we are not looking for query lower bounds, we are able to exploit this key idea for randomness
lower bounds. For any (even randomized) partition oracle, the Ramsey theory argument shows that existence
of some set of vertex labels such that algorithm can only compare these labels. This is exactly where we need
to choose vertex labels from a large enough set for the Ramsey theory to kick in. In turn, this leads us to the \emph{\ourmodel}, which we think is an interesting model in its own right. It leads one to wonder to what extent are the standard property testers reliant on the label set being $[n]$ (i.e. identical to the vertex set). We use the above fact to show that,
with limited randomness, the partition oracle must have the same behavior on various portions of the graph. With 
a careful accounting, we give a relation between the amount of randomness and the ability to only
cut $O(\eps n)$ edges. Overall, we can prove that with $\ll \lg n$ random bits, the algorithm is forced
to cut too many edges.

We consider this result to be a compelling conceptual contribution. It would be of interest to see
if this methodology can be used to other LCA applications. We note that our argument requires
the possible vertex labels to come from an astronomically large set. This leaves the following interesting open question.

\begin{problem} \label{prob:mainv2} Do there exist random seed length lower bounds for partition oracles, where the vertex set is fixed to be $[n]$?
\end{problem}

\subsection{Related Work} \label{sec:relatedv2}

In Chapter 9 of \cite{Goldreich:11}, Goldreich overviews \cite{HKNO} and provides a simple implementation of partition oracle where query complexity is exponential in $\poly(1/\eps)$. \cite{HKNO} observed that this primitive can be used to obtain testers all minor-closed properties on bounded degree graphs with two-sided error with query bound effectively being the same as the query complexity of the partition oracle implementation. In a series of works, \cite{KSS:18, KSS:19} sidestepped the partition oracle primitive and directly developed random-walk based testers for minor-freeness on bounded degree graphs.

In another important line of work, building up on techniques from \cite{HKNO}, Levi and Ron \cite{LR15} obtained an improved implementation of a partition oracle which runs in time quasi-polynomial in $1/\eps$. \cite{KSS:21} gave the first implementation of partition oracles which runs in time $\poly(1/\eps)$. \cite{LeviS21} repurposed the partition oracles from \cite{KSS:21} to give a polynomial time implementation of ``covering oracles'' and showed how this primitive can be used to efficiently test Hamiltonicity in planar graphs.

As noted earlier, all of the foregoing implementations of a partition oracle used $O(n \log n)$ bits of randomness. Thus, the task of setting up the partition oracle requires an intricate coordination among the vertices and this defeats the central point of LCAs \cite{AlonRVX12, RubinfeldTVX11}. On the algorithmic side, the central thrust of this work lies in managing the randomness complexity. 

On the lower bound side, which is the main contribution of the paper, the goal is to understand just how little randomness can one get away with when implementing a partition oracle with low query-complexity. Motivated by the original insights of \cite{Fischer04, ChakrabartyS13}, the central idea here is to use Ramsey Theory arguments which allow us to consider a \emph{canonical} implementation of partition oracles (which can \emph{simulate} any partition oracle with at most a polynomial overhead). Lower Bound arguments with roots in Ramsey Theory have been used since in Distribution Testing literature as well \cite{DKN15, CanonneW20, pinto23}. 

\section{The Lower Bound} \label{sec:lb:v2}

We first describe the \emph{\ourmodel}. 

\begin{definition} \label{def:huge} 
    The \ourmodel decouples the set of vertices of the input graph and the set from which labels for the vertices are chosen. The label of a vertex $\ell(v)$ comes from $[N] \subset \NN$ where the value $N$ is known to the algorithm under this model. Additionally, the algorithms also know $n = g^{-1}(N) \ll N$ (for a sufficiently slowly growing function $g$) and $d$. Here, $n$ denotes the number of vertices in the input graph and $d$ is an upperbound on the maximum degree. The \ourmodel also supports the following queries.

\begin{itemize}
\item Label queries: given an index in $i \in [n]$, output the $i$th vertex label, (eg, output $\ell(i)$) among an arbitrary order of all vertices.
\item Neighbor queries: given vertex $v$ and index $r \leq d$, return the (label of) the $r$th neighbor of the vertex $v$.
\end{itemize}
\end{definition}

\medskip

Note that this model is more challenging, since the partition oracle gets its input from $[N]$. If the input is (say) some natural number $v$, then unlike with the standard property testing model, the partition oracle has little idea about what the remaining labels could be. Nevertheless, intuitively, the partition oracle had better have a well-defined behavior for all possible inputs. 

In the rest of our discussion, the input graph $G$ is fixed to be an $n$-cycle.
The only variable (for the algorithm) is the labeling function $L: V \to \NN$.
We use the term ``vertex label", rather than ``vertex", to emphasize this fact.
The lower bound idea shows that achieving the global consensus from a local algorithm
requires $\Omega(\log n)$ bits. So the difficulty is not discovering the graph structure per se, but
rather in agreeing upon the cut edges. We refer to a vertex adjacent to a cut edge as a \emph{cut vertex}. Note that in an $n$-cycle a cut-edge results in two cut-vertices at either ends.

\subsection{Decision Tree Partition Oracles} \label{sec:canon}

With all this terminology in place, we now describe a partition oracle in terms of a ``cut function". Instead of thinking of the explicit connected components
that a partition oracle outputs, we will ask for a function that just determines which vertices are cut. For the cycle graph, it is easy to reconstruct the component a vertex belongs to by continuously querying neighbors and checking if they are cut. This does not affect any of the bounds by more
than a constant factor, and eases the presentation. Furthermore, it will be convenient
to have a single parameter that encompasses the query complexity and the size bound.

\begin{definition} \label{def:cut} A \emph{$q$-query} cut function partition oracle
is an algorithm that takes as input a vertex label $\ell$ and outputs a boolean value, where a value of $1$ represents the vertex being cut.
The algorithm makes at most $q$ queries to the graph and guarantees that all connected components have size at most $q$.
\end{definition}

(Note that typically, both $q$ and the component size are some functions of $\eps$. In this definition,
we simply bound the component size by the query complexity.)
Henceforth, when we refer to a partition oracle, we mean the cut function representation. So we drop the ``cut function"
terminology.
Note that query constraints must hold for all vertex labels and all possible labelings of $G$. 

We now define a sort of canonical partition oracle. The algorithm maintains a set of seed labels, initialized
to the input label. For each seed label $\ell$, the algorithm makes repeated neighbor queries to determine all labels
of vertices at distance at most $q$ from $L^{-1}(\ell)$. Based on all of these labels and the subgraph induced
by the corresponding vertices, the algorithm computes an index $i$. It then makes a label query to get a new seed label,
and repeats the above. This process is repeated $q$ times, after which the algorithm gives its output.
We refer to such a procedure as a canonical $q$-query partition oracle.

Observe that a canonical $q$-query partition oracle makes at most $O(q^2)$ queries; there are a total
of $q$ seeds, and from each seed, the oracles makes at most $2q$ neighbor queries. Note that subgraph
induced by the queried vertices is a collection of disjoint segments.

We now represent a canonical partition oracle as a decision tree $T$. Each node $z$ contains a function,
and the edges contain indices in $[n]$. (We use the word ``contains" instead of ``labeled", to avoid
confusion with the vertex labels of the graph.) Let $d_z$ denote the depth of node $z$, where
the root has depth $1$. The node $z$ contains a function $$f_z: \NN^{d_z(2q+1)} \to [n].$$

We think
of the input to $f_z$ as $d_z$ segments, each of length $2q+1$. The $r$th segment has the vertex labels
of the interval (in $G$) of length $2q+1$, corresponding to all vertices at distance at most $q$
from the $r$th seed vertex. The output of this function is the index for the next label query.

For clarity, let us explain how the overall output is generated. For a label $\ell$, let $B(\ell)$
denote the list of labels of all vertices at distance at most $q$ from (the vertex corresponding to) $\ell$.
The root node $y$ takes in the input
vertex label $\ell$. The partition oracle performs neighbor queries to get all labels $B(\ell)$. These
$2q+1$ labels are the arguments to the function $f_y$, which gives the index for the next label query.
The partition oracle determines this vertex label and proceeds. In general, consider a tree node $z$
at depth $d_z$. There are $d_z$ seeds that have been collected up to this point. For each such seed $\ell$,
the partition oracle determines the lists $B(\ell)$. Each of these, in order, forms the arguments
to the function $f_z$, which outputs the index of the next seed.

Note that the tree has depth at most $q$ and each node has fanout at most $n^{O(q^2)}$ (the number of possible indices).
So the overall representation is polynomial in $n$, even though the number of queries made is $q$ (constant). The key takeaway from this discussion is written in the following observation.

\begin{observation} \label{obs:trees-are-funcs}
    Fix a $q$-query canonical partition oracle represented as a decision tree $T$. Then $T$ can be alternately represented by a finite collection of $O(q^2)$ functions $f_z \colon \NN^{d_z(2q+1)} \to [n]$.
\end{observation}

\emph{Randomized partition oracles:} A randomized partition oracle $\bT$ using $r$ random bits is a collection
of $2^r$ decision tree representations $\bT = \{T_1, T_2, \ldots, T_{2^r}\}$. One of these is chosen
uar before the algorithm starts.

\Thm{main:v3} from the introduction is essentially a restatement of the following lower
bound theorem. Note that we think of $q$ as constant, with respect to the
number of vertices $n$.  

\begin{theorem} \label{thm:lb:v2} Consider a randomized $q$-query partition oracle $\bT$ using $r$ random bits in the \ourmodel with label size $N$.
For any $r = O(1)$, the following happens. There is an input labeling of $G$
such that, with probability $1$ (over the choice in $\bT$) at least $0.99n$ vertices are cut.
\end{theorem}

\subsection{Converting to Comparison-based Algorithms} \label{sec:ramsey}

Towards proving a lower bound on randomness complexity of a partition oracle with small query complexity, one might begin with the following reflex: For $i < Q$, the $i$-th query issued by a low-query (canonical) partition oracle is some function of the preceding label queries or the neighbor queries. Intuitively, one might suspect that this function depends only on the relative order among the preceding queries. We make this instinct precise by using arguments from Ramsey Theory.

We use $\cS_k$ to denote the symmetric group on $k$ elements, which is the set of all permutations of length $k$. We start with a definition.

\begin{definition} \label{def:comp} Given a list of $k$ natural numbers $(x_1, x_2, \ldots, x_k)$,
$\pi(x_1, x_2, \ldots, x_k)$ is the permutation induced by this list.

A function $f: \NN^k \to [n]$ is called \emph{comparison-based} if there exists $g:\cS_k \to [n]$
such that $f(x_1, x_2, \ldots, x_k) = g(\pi(x_1, x_2, \ldots, x_k))$.
\end{definition}

For a comparison-based function, only the ordering of the arguments matter, not their actual values. 
For convenience, given a function $f: \NN^k \to [n]$ and a subset $S \subseteq \NN$, let $f|_S$ be the function
obtained by restricting the inputs to $S^k$.

To introduce the Ramsey phenomena we are interested in, let us consider the following definition.
\begin{definition}
Let $s,t,w \in \mathbb{N}$. We define $r_t(s,w)$ to be the smallest integer $N$
such that every $w$-coloring of the $t$-tuples of $[N]$ contains a subset
$C \subseteq [N]$ of size $s$ for which all $t$-tuples from $C$ receive the same
color.
\end{definition}

\begin{fact}[\cite{ErdosR52} and Cor 10.1.5 in \cite{W24-Ramsey}]\label{fact:ramsey}

We have $r_t(s,w) = 2^{2^{2^{\cdot^{\cdot^{\cdot^{\cdot^{\cdot^{2^{(C_t w \log w) s}}}}}}}}} 
$
\end{fact}

\begin{claim}[Simultaneous Ramsey theorem]\label{clm:SRT}
Let $t,s,w,t',s',w' \in \mathbb{N}$, and define
$c := \max\{s,s'\}.$ 
Then there exists
$$N \leq r_t\!\big(r_{t'}(c,w'),\, w\big) \leq r_a\!\big(r_a(c,b),\, b\big) \text{ where } a = \max(t,t'), b = \max(w,w')$$ 
such that for every pair of colorings
$$
\chi : \binom{[N]}{t} \to [w],
\qquad
\chi' : \binom{[N]}{t'} \to [w'],
$$
there exists a subset $X \subseteq [N]$ with $|X| = c$ satisfying
$$
\chi \text{ is constant on } \binom{X}{t}
\quad\text{and}\quad
\chi' \text{ is constant on } \binom{X}{t'}.
$$
\end{claim}

\begin{proof} 
    Follows as an immediate corollary of \Fact{ramsey}.
\end{proof}

We begin by restating a classic Ramsey Theory result from \cite{Fischer04, ChakrabartyS13}.

\begin{theorem} \label{thm:ramsey2} Consider a finite collection of functions $\{f_i\}$ where
$f_i: \NN^{k_i} \to [w_i]$ (for $k_i \in \NN$) (where all $w_i$'s equal $n$). There exists an infinite subset $\MM \subseteq \NN$ such that 
$\forall i$, $f_i|_\MM$ is comparison based.
\end{theorem}

Suppose our collection $\{f_i\}$ of functions holds $\alpha$ functions in total. We are now ready to present a finitary version of this Ramsey Theorem. It asserts that one may assume that for every $i$, domain of $f_i$ is $k_i$-tuples all of whose elements come from $[N]$ for some finite $N$ and additionally, the set $\MM \subseteq [N]$ and thus $\MM$ is finite as well. Write $\max k_i \leq t$ and note $\forall i, w_i = w = n$. Also, note that for a fixed $i$, the function $f_i$ colors $k_i$-tuples of a ground set containing $n$ elements and this means we have $\forall i, s_i = s \leq n$.

\begin{theorem} \label{thm:finite-ramsey}
Consider a collection $\{f_i\}$ of functions where $f_i \colon \NN^{k_i} \to [n]$. Then there exists a finite $N \in \NN$ and a subset $C \subseteq [N]$ such that

\begin{itemize}
    \item $|C| \geq n$.
    \item For every $i$, the restriction of $f_i$ to $k_i$-tuples chosen from $C$ is constant.
\end{itemize}

\end{theorem}

\begin{proof}
    Let the size of the collection of functions $\{f_i\}$ be $\alpha$. For all $1 \leq i \leq \alpha$, set $s_i = s = n, w_i = w = n,  \max k_i = t$. Apply \Clm{SRT} repeatedly $\alpha$ times with parameters $s,t,w$. This nested application produces a value $N$ which respects the desired Ramsey phenomena simultaneously with respect to all the $f_i$'s. 
\end{proof}

The foregoing result, \Thm{finite-ramsey} upgrades from an infinite Ramsey-style theorem to a finite Ramsey-style theorem and it allows us to talk about functions where the set from which labels are chosen for the vertices in the \ourmodel can be bounded. Next, we show using a low randomness implementation of $\bT$, the number of label queries that $\bT$ ``\emph{can ever issue}'' can be shown to be no more than $n/q^5$. We emphasize that we are not merely bounding the total number of label queries issued by the canonical oracle during a single execution. Whereas, we are bounding the total number of label queries $\bT$ issues across all of its unrolled execution tree, which is potentially a much bigger quantity.

\begin{claim} \label{clm:indexv2} 
    Take $n \in \NN$ and fix an $n$-cycle, $G$. Consider a randomized $q$-query partition oracle $\bT$ using $r < \lg n - q^6$ random bits. Then there exists a set $C$ of labels such that if vertices of $G$ are labeled according to $C$, then the total number of indices used for label queries by $\bT$ (over all trees in $\bT$) is at most $n/q^5$.
\end{claim}

\begin{proof}
    Recall that a randomized partition oracle $\bT$ using $r$-bits of randomness picks a uniformly random decision tree from a collection of decision trees $\{T_1, T_2, \ldots T_r\}$. Consider the canonical representations for all these trees and note that by \Obs{trees-are-funcs}, we can view $\bT$ as a finite list of functions of the form $f_i \colon \NN^{k_i} \to [n]$. Note that $\max k_i \leq O(q^2)$ and that all of these functions can be thought of as colorings with a palette of size $w_i = n$. By \Thm{finite-ramsey}, for a large enough $N$ depending on all these parameters, there exists a subset $C \subseteq [N]$ with $|C| = n$ such that the restriction of $f_i$ to $k_i$ tuples from $C$ is a constant. Thus, we note that all these functions, when restricted to labels from $C$ are comparison-based (see \Def{comp}). Let us focus on labelings to $G$ that come from this set, $C$. \\ 

    Now, consider the finite family of \emph{all} functions over \emph{all} nodes in \emph{all} trees of $\bT$.
    Take any tree $T \in \bT$ and any node $z \in T$ and the corresponding function $f_z$.  Since $f_z$ has at most $q(2q+1)$ arguments and is comparison-based, it can have at most $(q(2q+1))! \leq q^{5q^2}$ outputs. Hence, each node has fanout at most $q^{5q^2}$ (instead of $n^{O(q^2)}$ for general input labelings). Hence, restricted to inputs from $\MM$, any tree in $\bT$ has size at most $(q^{5q^2})^q = q^{5q^3}$. The total number of edges in all trees in $\bT$ is at most $$ 2^r q^{5q^3} \leq 2^{\lg n - q^6} q^{5q^3} \leq n/q^5 $$ 
    
    Observe that any index used for a label query is on some edge of a tree in $\bT$. Hence, there are at most $n/q^5$ indices used by $\bT$, when the input labels are restricted to $C$. 
\end{proof}

\begin{remark} \label{rem:inverse-ramsey}
In the \ourmodel, as noted in \Clm{indexv2}, for a fixed $n$, the label set size $N$ satisfies $N \geq |C_{\ref{thm:finite-ramsey}}|$ which means $N$ has a wowzer-like growth in $n$. Thus, for a fixed $N$, under the \ourmodel, we may assume $n \ll g^{-1}(N)$ for a sufficiently fast growing function $g$.
\end{remark}

\subsection{Proof of \Thm{lb:v2}} \label{sec:lb-proof}
We have all the tools to complete the proof. Consider the set $C$ from $\Clm{indexv2}$. 
We will label $G$ with the elements in $C$ in order. So the vertex labels
in $G$ will be $m_1 < m_2 < m_3 < \ldots < m_n$, where each $m_i \in C$. 

Since the labeling is fixed, we will simply refer to a vertex by its label.
Consider the set of at most $n/q^5$ indices obtained from \Clm{indexv2}. These correspond
to specific vertices in $G$, which we call the entire seed set. A seed vertex refers
to an element in this set.

\begin{definition} \label{def:chunk} A \emph{chunk} is a maximal contiguous interval
of at least $q^2$ vertices that are at distance at least $q$ from every vertex
in the entire seed set. Moreover, all labels in a chunk are in monotonic order
(as we walk along the interval).
\end{definition}

A simple calculation shows that an overwhelming majority of vertices are in chunks.

\begin{claim} \label{clm:chunk} At least $(1-1/q^2)n$ vertices are in chunks.
\end{claim}

\begin{proof} We add the lowest label and the largest label (which are neighbors) to the seed set.
We delete all vertices within distance $q$ of any seed vertex in $G$.
By \Clm{indexv2}, there are at most $n/q^5$ seeds, so a total of $(n/q^5)(2q+1) \leq n/q^3$
vertices are deleted. Since we deleted at most $n/q^5$ contiguous segments,
there are at most $n/q^5$ connected components remaining. Each of these has labels in monotonic order.
Let us now delete
any connected component of size at most $q^2$, so an additional (at most) $n/q^3$ vertices
get deleted. All remaining connected components are chunks. There are at least
$n - 2n/q^3 \geq (1-1/q^2)n$ vertices in chunks.
\end{proof}

The next claim argues that the randomized partition oracle $\bT$ cannot distinguish vertices
within a chunk.

\begin{claim} \label{clm:dist-chunk} Consider a chunk $C$. For all $T \in \bT$,
the output of $T$ on all input labels in $C$ is identical.
\end{claim} 

\begin{proof} Consider two labels/vertices $m, m' \in C$. Let $S$ denote the entire seed set.
The ordering within
$B(m)$ is just monotonic order, and is the same for $B(m')$. Since $m$ is in a chunk,
$B(m)$ is disjoint with $B(s)$ for all $s \in S$. 
Consider the ordering of the elements of $B(m)$ with respect to $\bigcup_{s \in S} B(s)$. Observe
that this ordering is the same whether we choose $B(m)$ or $B(m')$. 

In other words, take any node $z$ of a decision tree in $\bT$. The arguments of corresponding
function $f_z$ involve $B(s)$ (for some seeds in $s \in S$) and $B(x)$ where $x$ is the input.
All of these values induce the same permutation regardless of whether the input $x$ is $m$
or $m'$. Hence, the output of all nodes in all decision trees of $\bT$ is identical,
implying that the overall output is identical for all inputs from $C$.
\end{proof}

We wrap up the proof. Observe that, for any chunk $C$, some vertex must be cut,
otherwise there is a connected component of size at least $q$. By \Clm{dist-chunk},
all vertices in a chunk are cut, by every partition oracle in $\bT$. And by \Clm{chunk}, at least $(1-1/q^2)n$ vertices are cut by every partition oracle, completing the proof of \Thm{main}.

\section{Partition Oracle in the \texorpdfstring{$b$}{b}-wise Independent Framework}
In this section, we prove \Thm{main}. We prove the slightly stronger version for \ourmodel.

\begin{theorem} \label{thm:main:huge-label}
Consider the \ourmodel with label set of size $N$. Let $\cP$ denote the set of $d$-bounded degree graphs in a minor-closed family on $n \leq N$ vertices. Then there exists an $(\eps, \poly(d/\eps), \poly(d/\eps) \cdot \log N)$-partition oracle for $\cP$.
\end{theorem}

We begin by revisiting the core partitioning framework introduced by \cite{KSS:21}, which this paper reuses in full. We note that this algorithm can be extended easily to the \oursetting.

We note that the novelty of our work does not lie in proposing a new algorithm, but rather in demonstrating that their algorithm can be implemented with dramatically reduced randomness complexity. This section aims to establish notation, recall key subroutines from \cite{KSS:21}, and highlight where randomness enters. We also describe how we reorganize the randomness into structured sources that are amenable to derandomization.

\subsection*{Global Diffusion Framework and Key Parameters}
The algorithm in \cite{KSS:21} is built on a global diffusion process over a bounded-degree graph $G = (V,E)$ with $n$ vertices and maximum degree $d$. All relevant parameters are set as functions of a proximity parameter $\eps > 0$:

\parameterinfo

Random walks are governed by the symmetric lazy random walk matrix $M$, where
\[
M_{u,v} = \begin{cases}
\frac{1}{2d} & \text{if } (u,v) \in E, \\
1 - \frac{d(u)}{2d} & \text{if } u = v, \\
0 & \text{otherwise}.
\end{cases}
\]

We recall the truncated diffusion operator, a central primitive:

\begin{definition}[Truncated Diffusion \cite{KSS:21}]\label{def:trundiff}
Let $\vec{x} \in (\RR^+)^n$. Define $\trwalk(\vec{x})$ to be $M \vec{x}$ with all coordinates $\le \minp$ zeroed out. The $t$-step truncated diffusion is defined recursively as $\trwalk^t(\vec{x}) := \trwalk(\trwalk^{t-1}(\vec{x}))$, with $\trwalk^0(\vec{x}) := \vec{x}$.
\end{definition}

We write $\level{v}{t}{k}$ to denote the \emph{level-set} of the $k$ largest coordinates (by value) in $\trwalk^t(\vec{v})$, breaking ties lex in vertex IDs. Here, $\vec{v}$ denotes the indicator vector with support $v$.

For local simulation and charging arguments, we will also need:

\begin{definition}[Inverse Ball \cite{KSS:21}]
\label{def:ball}
For any $v \in V$, define the inverse ball:
\[
\ball(v) := \left\{ w \in V \mid \exists t \in [0,\ell] \text{ such that } v \in \supp(\trwalk^t(\vec{w})) \right\}.
\]
\end{definition}

\begin{claim}[\cite{KSS:21}]
\label{clm:ballsize}
For every $v \in V$, $|\ball(v)| \leq \ell / \minp$.
\end{claim}

\subsection*{The Clustering Subroutine}

The central algorithmic subroutine in this framework is a local routine that attempts to extract a low-conductance cut near a vertex $v$. This subroutine is reproduced from \cite{KSS:21} for completeness.

\medskip
\noindent
\clustercode

\medskip
The inputs $t_v$ and $k_v$ are drawn from random sources. Our derandomization effort focuses on simulating these with pseudorandom distributions while maintaining correctness guarantees. We note that the set $L_{v,t,k}$ can in fact be efficiently computed. Indeed, suppose we know the values $t_v$ and $k_v$ (for details on how to do this, see \Sec{globalg}). The remaining algorithmic work in computing this level set is done by performing truncated diffusions repeatedly for $t_v$ steps. One needs only track all the vertices which are reachable with at least some $\minp = \poly(\eps)$ probability. Among these vertices, one just picks the top $k$ and this is the level set $L_{v,t,k}$.

\subsection*{Global Algorithm and Randomness Organization} \label{sec:globalg}
We now describe the global clustering algorithm (originally from \cite{KSS:21}) with our new randomness bookkeeping. In our implementation, the randomness is split into two sources:

\begin{itemize}
  \item  $\bS_1$, which is a random seed fed to the pseudorandom generators for producing the phase assignments $h_v$'s and diffusion lengths $t_v$'s; 
  \item $\bS_2$, which contains the small amount of randomness used during preprocessing (to compute $k_v$ via \findr).
\end{itemize}

We view $\bS_1$ as the concatenation of $\overline{h}$ small seeds:
\[
\bS_1 = \bS_1(1) \circ \bS_1(2) \circ \cdots \circ \bS_1(\overline{h}).
\]
Here, $\overline{h}$ is some value which is at most $ \leq 2/\delta$ and it denotes a ``phase-cutoff'' parameter -- that is, we focus on the first $\overline{h}$ phases as intuitively, there are not too many vertices after this phase. For each $h \leq \overline{h}$, $\bS_1(h)$ is used to generate the following two $b$-wise independent collections over $V$.

\begin{remark} \label{rem:s1-randomness}
From Construction 3.32 in \cite{Vadhan12}, the following bounds are known on the length of the random seeds, $\bS_1(1)$ through $\bS_1({\overline{h}})$ used to generate $\{H_v\}$'s and $\{t_v\}$'s

\begin{itemize}
  \item $\{ H_v^{(h)} \sim \texttt{Ber}(\delta) \}_{v \in [N]}$: to decide if $v$ participates in phase $h$. Across all the $\overline{h}$ phases, producing this collection consumes a total of $O(\overline{h} b \cdot \log N)$. \label{itm:phasebits}
  \item $\{ T_v^{(h)} \sim \texttt{Unif}([1,\ell]) \}_{v \in [N]}$: diffusion length for $v$ in phase $h$. Across all the $\overline{h}$ phases, producing this collection consumes a total of $O(\overline{h} b \cdot \log \ell \cdot \log N)$. \label{itm:timestepbits}
\end{itemize}
\end{remark}

Given these collections, each index $v \in [N]$ sets its phase $h_v := \min \{ h : H_v^{(h)} = 1 \}$ and its diffusion length $t_v := T_v^{(h_v)}$.

The final piece is the conductance threshold $k_v = k(h_v)$, which is computed using the \findr routine. Since this procedure already has low randomness requirements, we abstract it as depending only on $\bS_2$.

We now reproduce the full global partitioning algorithm, originally from \cite{KSS:21}, but updated to reflect our structured randomness sources $\bS_1$ and $\bS_2$. This algorithm will later be locally simulated.

\medskip
\noindent
\globpartcode

\medskip
Our contribution is the analysis showing that it can be implemented using only $O(\delta^{-1} \log N)$ bits of randomness via $b$-wise independent sources and suitable pseudorandom generators. We note that in the standard model we use only $O(\delta^{-1} \log N)$ random bits.

\subsection{The \findr Subroutine and the Role of $k_v$} \label{sec:findrmatters}

As described earlier, the conductance thresholds $\{k_v\}$ are computed via the \findr subroutine. While the subroutine itself is not changed, we summarize its purpose here. First, a little notation. At the start of phase $h$, the set of unclustered vertices will be denoted $F_h$. The current free set will be denoted as $F$. The set of vertices with phase value being $h$ is denoted $V_h$, and the set of vertices with phase value at least $h$ is denoted $V_{\geq h}$.

\paragraph{Why \findr Matters.}
As the partitioning proceeds, new clusters may overlap heavily with previously clustered vertices. In such cases, $\cluster(v)$ might have low conductance in $G$ but high conductance in the current free set $F$. This could result in poor amortization of edge cuts. To mitigate this, \cite{KSS:21} show that ``good'' clusters occasionally arise—ones that add many new vertices while cutting few edges. The size thresholds $k_h$ help identify these.
Fortunately, the randomness used by \findr is already limited: it samples only a small number of vertices and evaluates their diffusion vectors. Thus, we treat its randomness usage as encapsulated in $\bS_2$. We finish with a high level sketch of the \findr procedure. \findr proceeds phase-by-phase. In each phase $h$, it:
\begin{itemize}
  \item Samples a small set of vertices from $V_{\geq h}$;
  \item For each such vertex $v$ and each $t \in [\ell]$, checks whether $\level{v}{t}{k}$ satisfies:
    \begin{itemize}
      \item (i) it's a level set in $\trwalk^t(\vec{v})$, 
      \item (ii) $\Phi(\level{v}{t}{k} \cup \{v\}) < \phi$, 
      \item (iii) $|\level{v}{t}{k} \cap F|$ is large.
    \end{itemize}
  \item Picks the smallest $k$ for which enough vertices satisfy all three.
\end{itemize}

Thus, in each phase \findr algorithm computes a value $k_h$ which is nice enough in the sense that enough vertices in the set satisfy that diffusions of many different lengths reveal a cluster containing enough free vertices. This is captured in the following definition.

\begin{definition} \label{def:viable}
    We call a vertex $v$ with phase value $h$ or larger as $(h,k)$-viable if there are at least $\frac{\sizefrac}{\log^2(1/\sizefrac)} \cdot \ell$ many choices of timesteps $t_v$ from $v$ which satisfy $|\mathtt{cluster}(v,t_v,k) \cap F| \geq \sizefrac^3 k$.
\end{definition}

Armed with this definition, we can now reproduce the $\findr(\bS_2)$ procedure from \cite{KSS:21}.

\medskip
\noindent
\findrcode

\medskip
The structural guarantee below (from \cite{KSS:21}) is the key tool enabling this subroutine.

\begin{theorem}[Isoperimetric Guarantee from \cite{KSS:21}]
\label{thm:enough_h_t}
If $|F_h| \geq \beta n$, then there exists a threshold $k \leq \rho$ such that at least $(\beta^2 / \log^2(1/\beta)) n$ vertices in $F_h$, there are at least $(\beta / \log^2(1/\beta)) \cdot \ell$ values of $t \in [\ell]$, such that there exists $k'$ between $k$ and $2k$ for which the set $\level{v}{t}{k'}$ satisfies all three of the conditions below: 
\begin{itemize}
      \item it's a level set in $\trwalk^t(\vec{v})$, 
      \item $\Phi(\level{v}{t}{k'} \cup \{v\}) < \phi$, 
      \item $|\level{v}{t}{k'} \cap F| \geq \sizefrac^3 k$
    \end{itemize}
\end{theorem}

\subsection{Two Core Theorems in the $b$-wise Setting} \label{sec:corethms}

The rest of the paper is devoted to verifying that the guarantees of \cite{KSS:21} continue to hold when the randomness is replaced by limited-independence sources. In particular, we prove two structural theorems:

\begin{restatable}{theorem}{thmourfindr}
\label{thm:findrv2}
With probability at least $1 - \exp(1/\eps)$ over the choice of randomness in $\bS_1 \circ \bS_2$, the following holds for all phases $h \le \overline{h}$. If $|F_h| \ge \sizefrac n$, then at least $\sizefrac^5 \delta n$ vertices in $V_h$ are $(h,k_h)$-viable.
\end{restatable}

\begin{restatable}{theorem}{thmedgecutbnd}
\label{thm:edgecut:bnd}
The expected number of edges cut by the global partitioning procedure $\globpart(\bS_1 \circ \bS_2)$ is at most $\varepsilon nd$.
\end{restatable}

Together, these results yield a randomization-reduced implementation of the partition oracle construction in \cite{KSS:21}, with total randomness usage reduced to $O(\delta^{-1} \cdot \log n)$ bits. 

\begin{remark}\label{rem:remainingtasks}
    The rest of the properties of partition oracle listed in \Def{oracle} carry over. In particular, all the pieces are found by truncated diffusions (\Def{trundiff}). So, no set can have size bigger than $\ell/\rho$. Next, $\globpart(\bS_1 \circ \bS_2)$ adds connected components of $\cluster(v) \cap F$ to the partitioning it computes. Thus, it is indeed a bonafide partition oracle. For details, see Theorem 3.12 in \cite{KSS:21}. Note that truncated diffusions are deterministic and therefore \cite{KSS:21}'s analysis requires no additional guarantees from the randomness of $\bS_1 \circ \bS_2$ to hold.    
\end{remark}

\subsection{Tail Bounds under Limited Independence}
\label{sec:tail-bounds}
To analyze the behavior of our pseudorandom simulation, we rely on concentration bounds for limited-independence random variables. These allow us to work with $b$-wise independent distributions rather than full randomness, while still retaining strong probabilistic guarantees.

We first recall a standard bound (see \cite[Problem 3.8]{Vadhan12}):

\begin{fact}[Chernoff Bound for $b$-wise Independent Bernoulli Variables]
\label{fact:kwise:chernoff}
Let $X_1, \dots, X_T$ be $b$-wise independent Bernoulli random variables with success probability $\delta$. Let $X = \sum X_i$. Then for any $\eps > 0$,
\[
\Pr\left[|X - \delta T| \ge \eps T \right] \le \left( \frac{b^2}{4 T \eps^2} \right)^{b/2}.
\]
\end{fact}

An immediate corollary controls the size of subsets selected by such variables:

\begin{corollary} \label{cor:kwise}
Fix $h < \overline{h}$ and a subset $F \subseteq V_{\ge h}$ of size at least $c n$. Then with probability at least $1 - \left( \frac{b^2}{c n \delta^2} \right)^{b/2}$, the number of vertices in $F$ assigned to phase $h$ satisfies
\[
\left| |V_h \cap F| - \delta \cdot c \cdot n \right| \le \delta \cdot c \cdot n / 2.
\]
\end{corollary}

\begin{proof}
    This follows by letting $T = c n$ and $\eps = \delta/2$ in \Fact{kwise:chernoff}.
\end{proof}

These bounds will be repeatedly used to ensure that (i) phases contain enough vertices, (ii) sampling succeeds with high probability, and (iii) viable clusters occur frequently enough to support our cut bound argument.

Additionally, we will need the following guarantees about sampling from the $\bS_1$ part of the seed. Abusing notation, for a random-variable $\bS$ which takes on value over a set of strings of small size, we will use the notation $\psg(\bS)$ to refer to the distribution over $\poly(n)$-length strings obtained by running the pseudorandom generator on the string $\bS$.

\begin{claim}
    Fix an index $u \in [N]$. The probability that a random vertex sampled using the pseudorandom string $\bR_1 \sim \psg(\bS_1)$ satisfies the following:

    $$\mypr_{s_1 \sim \bR_1}\big(s = u\big) = 1/N,$$
    and for any fixed tuple of $t$ vertices $u_1, u_2, \ldots u_t$, we have
    $$\forall t \leq \alpha, \mypr_{s_1, s_2, \ldots s_t \sim \bR_1}\big(\forall i \leq t, s_i  = u_i\big) = 1/N^t.$$
\end{claim}

\begin{fact} 
\label{fact:phase-prob} 
Fix $1 < h < \overline{h}$, for any $v \in V$, $\mypr[v \in V_h \ | \ v \in V_{\geq h}] = \delta$.
\end{fact}

\section{Proof of \texorpdfstring{\Thm{findrv2}}
{Theorem~\ref{thm:findrv2}}}

\thmourfindr*

\begin{proof}[Proof of \Thm{findrv2}]
We will mimic the proof of Theorem 4.10 from \cite{KSS:21} and, similar to that proof, we have two steps. \\
\noindent \textbf{Step 1:} We will show if $|F_h| \geq \beta n$ then, with high probability, $\findr(\bS_2)$  returns a non-zero value for $k_h$.
    
\noindent \textbf{Step 2:} Next, we will show that with high probability, if a non-zero $k_h$ is returned, then $V_{h}$ contains at least $2\beta^5 \delta n$ viable vertices, which will finish the proof.
    
    Fix an $h$. As the premise asserts, let us consider any choice of randomness in the preceding $h-1$ phases such that $|F_h| \geq \beta n$. We also note that $V_{\geq h} \supseteq F_h$ and therefore $|V_{\geq h}| \geq |F_h| \geq \beta n$. \\
    
    \noindent \textbf{Proof of Claims laid out in Step 1:} We claim that with high probability over the randomness in $\bS_2$, a call to $\findr(\bS_2)$ successfully processes line 1.(d). To establish our claim, we track some bad events.  First, we will show an easier claim, $\texttt{findr}(\bS_2)$ reaches Step 1.(c) with a high probability. Then we will show $\texttt{findr}(\bS_2)$ reaches Step 1.(d) with a high probability. In the end, we will argue that the $k_h$ returned by $\texttt{findr}(\bS_2)$ is positive with high probability. \\
    
    \noindent\textbf{Error 1: The probability $\texttt{findr}(\bS_2)$ fails to reach Step 1.(c) is low} 

    Note that if we pick a vertex $s$ from $V(G)$ using the pseudorandom generator, we have 
    $$\mypr_{\bS_2}(s \in S_h) = \frac{|V_{\geq h}|}{|V|} \geq \beta.$$
    This means, $\myex[|S_h|] \geq \beta^{-9}$.    

    Now, let us control the probability of error 1. So, we use Chernoff Bound to obtain. 
    
    $$\mypr_{\bS_2} \Big(|S_h| \leq \myex[|S_h|]/2 \Big)	 \leq \exp(-\beta^{-9}/8) \leq \exp(-\beta^{-8}).$$

    This bound holds because all the vertex samples in $S_h$ are obtained by using different random seeds, one for each $h$. This means the algorithm reaches Step 1.(c) with probability at least $1 - \exp(-\beta^{-8})$. \\
    
    \noindent\textbf{Error 2: The probability $\texttt{findr}(\bS_2)$ fails to reach Step 1.(d) is low}\\
    We know from \Thm{enough_h_t} that many candidate $k$'s exist, and for each $k$, there exist many $(h,k)$-viable vertices in $S_h$. In the \texttt{for} loop on Step 1.(c), $\texttt{findr}(\bS)$ loops over all eligible choices for $k$ to find a ``\emph{good}'' $k$. 
    Let us fix a value of $k$ as suggested by \Thm{enough_h_t}. 
    We have, 
  	\begin{align*}
     \mypr_{\bS_2}\Big(s \in V_{\geq h} \text{ and marked } (h,k)\text{-viable}\Big) & \geq  \mypr_{\bS_2}\Big(s \in V_{\geq h}\Big) \times \mypr_{\bS_2}\Big(t_s \sim [1, \frac{\beta}{\log^2 \beta^{-1}}] \Big)\\
     & = \frac{\beta^2}{\log^2 \beta^{-1}} \times \frac{\beta}{\log^2 \beta^{-1}} = \frac{\beta^3}{\log^4 \beta^{-1}}
    \end{align*}

    \noindent Hence, $$\myex[|s \in V_{\geq h} \text{ and marked } (h,k)\text{-viable}|] \geq \frac{\beta^3}{\log^4 \beta^{-1}}|S_h|.$$
    We use Chernoff Bound once more to bound the probability of discovering at most $12 \beta^4|S_h|$ many $(h,k)$-viable vertices.
	
	\begin{align*}
	\mypr_{\bS_2}\Big(X \leq 12 \beta^4|S_h|\Big) & \leq \exp{\Big(-(12\beta\log^4 \beta^{-1})^2\frac{\beta^3}{2\log^4 \beta^{-1}}|S_h|\Big)} & \text{\quad using } \mypr\Big(X \leq (1-\delta)\myex[X]\Big) \\
    & & \leq \exp\Big(- \frac{\delta^2\myex[X]}{2}\Big)\\
	& = \exp{(-72\beta^5\log^4 \beta^{-1}|S_h|)}
	\end{align*}
	
	This implies we explore at least $12 \beta^4|S_h|$ $(h,k)$-viable vertices with probability at least $1 - \exp{(-72\beta^5\log^4 \beta^{-1}|S_h|)}$. 
	Subsequently, the algorithm sets $k_h$ to a non-zero value after reaching Step 1.(d). 
	We can estimate the probability of the event that the algorithm reaches Step 1.(d) is $1 - \exp{(-\beta^{-8})} - \exp{(-72\beta^5\log^4 \beta^{-1}|S_h|)} \geq 1 - \exp{(-\beta^{-1})}$. Hence, we can bound the probability that the algorithm doesn't reach Step 1.(d) by $\exp{(-\beta^{-1}) }$.  \\

    \noindent\textbf{Error 3: The probability $\texttt{findr}(\bS_2)$ fails to return a $k_h > 0$ is low}\\
    
    This holds trivially from the previous step, as we know from the structural Lemma that many candidate $k$'s exist, and for each $k$, there exist many $(h,k)$-viable vertices in $S_h$. In the \texttt{for} loop on Step 1.(c), $\texttt{findr}(\bS_2)$ loops over all eligible choices for $k$ to find a ``\emph{good}'' $k$. We already showed that the failure probability of that event is low. 
    
    \noindent \textbf{Step 2:} Now we will first argue that the number of $(h,k)$-viable vertices in $V_{\geq h}$ is at least $2\beta^5 n$ with high probability. Once that step is done, we will argue that the number of $(h, k_h)$-viable vertices in $V_h$ is at least $\beta^5\delta n$ with high probability.\\ 
    
    \noindent\textbf{Error 4: If $\texttt{findr}(\bS_2)$ returned $k_h > 0$, then the number of $(h,k)$-viable vertices in $V_{\geq h}$ is less than $2\beta^5 n$ is low} \\
    
    We prove this by showing the contrapositive. The algorithm processed line 1.(d) and which returned some value $k = k_h$. Suppose the number of $(h,k)$-variable vertices in $V_{\geq h}$ is at most $2\beta^5 n$. We will show that the event of line 1.(d) returning $k = k_h$ is very unlikely.
Using the lower bound $|V_{\geq h}| \geq |F_h| \geq \beta n$,
    \begin{align*}
    \myex\Big[|s \in S_h \text{ and } s \text{ is } (h,k) \text{-viable}|\Big] \leq \frac{2\beta^5n}{|V_{\geq h}|}|S_h|\leq  2\beta^4|S_h|.
    \end{align*}

    Let $X_k$ denote the random variable of the number of $(h,\sizethresh)$-viable vertices in $S_h$. Also, note that any sampled vertex in $S_h$ is a $(h,k)$-viable with some probability $p \leq 2\beta^4$ independent of all the other vertices sampled in $S_h$. Thus, by \Fact{kwise:chernoff}, $\mypr[X_k > 12\sizefrac^4|S_h|] \leq 2^{-12\sizefrac^4|S_h|}$.

    Hence, the probability of discovering $12\beta^4|S_h|$ vertices is low,  and therefore the call to $\findr(\bS_2)$ so is the probability that $\findr$ returned and $k_h$  be $k$. Contradiction. \\     
     
    \noindent\textbf{Error 5: If $\texttt{findr}(\bS_2)$ returned $k_h > 0$, then the number of $(h,k_h)$-viable vertices in $V_{h}$ is less than $\delta\beta^5 n$ is low} \\     
    
    From the previous step, we know that the number of $(h,k_h)$-viable vertices in $V_{\geq h}$ is at least $2\beta^5n$. Using \Cor{kwise}, it follows with probability at least $1-n^{-b/4}$, that the number of variable vertices in $V_h$ is at least $\frac{\delta}{2}|(h,k)\text{-viable vertices in } V_{\geq h}| = \delta \beta^5 n$. 

    Finally, we use a union bound over all the $\overline{h}$ phases and all the $1/\rho$ choices for $k$ and all the errors encountered thus far. The total error probability is at most $$\overline{h} \cdot \exp(-\sizefrac^{-8}) + \exp(-1/\sizefrac) + 2^{-12/\sizefrac^4} + n^{-b/4} \leq 4 \exp(-1/\sizefrac).$$
\end{proof}

\section{Proof of \texorpdfstring{\Thm{edgecut:bnd}}{Theorem~\ref{thm:edgecut:bnd}}}

In the standard model, \cite{KSS:21} employs a random seed $\bR$ of length $\frac{n}{\mathrm{poly}(\varepsilon)} \cdot \log n$.\footnote{In the \ourmodel, this takes $\frac{N}{\mathrm{poly}(\varepsilon)} \cdot \log N$ bits.} In their framework, a call to $\findpart(v, \bR)$ internally invokes $\findanchor(v, \bR)$, which in turn calls $\cluster(v)$.

In our setting, however, we invoke $\globpart(\bS_1 \circ \bS_2)$ using $\bS_1 \circ \bS_2$ as the randomness source. This introduces correlation between partitioning steps within $\globpart(\bS_1 \circ \bS_2)$ when it invokes two different $v$ and $u$. While \cite{KSS:21} proves that the expected number of cut edges remains small under the use of $\bR$, our use of $\bS_1 \circ \bS_2$ requires reproving the cut bound. In this section, we give a straightforward adaptation of techniques from \cite{KSS:21} to show a similar bound on the expected number of cut edges under our correlated randomness model. The remainder of this section is devoted to explaining how techniques from \cite{KSS:21} can be adapted to prove this bound. We first overview our proof of the cut bound. Following \cite{KSS:21}, we first express our bound on expected size of the cut set in terms of the expected value of following quantity $$\sum X_h = \sum_{v \in V_h} |\cluster(v) \cap F_h|$$ which measures the number of free vertices clustered at the beginning of phase $h$ with repetitions across all the phases. The key insight is we can mimic the expected value computations from \cite{KSS:21} because our random variables are $b$-wise independent. We start by recalling the main theorem we prove in this section.

\thmedgecutbnd*

We first collect some ingredients we need to prove this theorem. We start by re-purposing some basic results (done in \Prop{bwise:basic}) about a collection of $b$-wise independent random variables that we will use in our final proof. 
\begin{proposition} \label{prop:bwise:basic}
    Let $h^* = \lceil 2/\delta \rceil$ and $b = 4/\sizefrac^{10}$. With probability at least $1 - n^{-b/4}$, there exists a cutoff phase $1 \leq h_c \leq h^*$ such that $|F_{h_c}| \leq \sizefrac n$.
\end{proposition}

\begin{proof}
    We prove this result by contradiction. We will show that if $|F_{h^*}| \geq \sizefrac n/2$, then with high probability $|F_{h^* - 1}| \leq \sizefrac n/2$. Since the size of the free sets can only go down, the cutoff phase must be encountered somewhere between phase 1 and phase $h^*$ with high probability. 

    So, let us assume $|F_{h^*}| > \sizefrac n/2 \geq 10 \delta n$. Thus, all the preceding phases also contain at least $N = 10 \delta n$ vertices. Fix one of these preceding phases, say phase $h$. The set $V_h \subseteq V_{\geq h}$ is obtained by having each vertex toss a coin with heads probability $\delta$ and this entire collection of tosses is $b$-wise independent. Therefore, by \Cor{kwise}, we get

    $$Pr(|V_h| \geq \delta n/2) \leq 1 - \left(\frac{b^2}{N \delta^2}\right)^{b/2}.$$

    By a union bound over at most $h^*$ phases, we get $$Pr\Bigg(\forall h \leq h^* \colon |V_h| \geq \delta n/2 \Bigg) \geq 1 - h^* \cdot \left(\frac{b^2}{N \delta^2}\right)^{b/2} \geq 1 - \left(\frac{1}{\sqrt n}\right)^{b/2} = 1 - n^{-b/4}.$$
\end{proof}

This theorem follows by combining the following two claims.

\begin{claim} \label{clm:cutvalbnd}
    $$\EX[\# \text{ of edges cut by } \globpart(\bS_1 \circ \bS_2)] \leq 4096 \phi \sizefrac^{-8} d^2 \Bigg(\sum_{h \leq \overline{h}} \EX\Big[\sum_{v \in V_h} |\mathtt{cluster}(v) \cap F_h|\Big] \Bigg) +  \sizefrac n d/2.$$
\end{claim}

Consider the $h${th} phase. The quantity $\EX\Big[\sum_{v \in V_h} |\mathtt{cluster}(v) \cap F_h|\Big]$ counts (with repetitions) the total number of free vertices clustered during the $h${th} phase. Note that repeats are inevitable as it is possible for calls to $\cluster(u)$ and  $\cluster(v)$ (where $u,v \in V_h$) both to have some overlap within the set $F_h$. The claim below asserts that summed up over all the phases, this cannot happen too often and that it is much better than the naive bound $$ \sum_{h < \overline{h}} \EX[\sum_{{v \in V_h}} |\mathtt{cluster}(v) \cap \free{h}|)] \leq \sum_{h < \overline{h}} \sum_{v \in V_h} |\mathtt{cluster}(v) \cap F_h| \leq \sum_{v \in V} |\ball(v)| \leq \ell \rho^{-1} \cdot n.$$

\begin{claim}\label{clm:newcharging} 
$$ \sum_{h < \overline{h}} \EX[\sum_{{v \in V_h}} |\mathtt{cluster}(v) \cap \free{h}|)] \leq 4n $$
\end{claim}

We finish the proving \Thm{edgecut:bnd} first. Thereafter, we will prove the claims.

\begin{proof}[Proof of \Thm{edgecut:bnd}.] 

Take $\eps > 0$ to be sufficiently small and let us upper bound the number of edges cut by a call to $\globpart(\bS_1 \circ \bS_2)$. Following \Clm{cutvalbnd}, we have.

\begin{flalign*}
    \EX[\# \text{ of cut edges}] &\leq 4096n\phi \sizefrac^{-8} d^2 \Bigg(\sum_{h \leq \overline{h}} \EX\Big[\sum_{v \in V_h} |\cluster(v) \cap F_h|\Big] \Bigg) +  \sizefrac n d/2 \\
    &\leq 16,384 \phi\sizefrac^{-8}d\cdot nd + \sizefrac nd/2 &&\text{By \Clm{newcharging}} \\
    &\leq \eps nd &&\text{Using $\phi = \frac{\eps^{10}}{d}, \sizefrac = \eps/10$}\\
\end{flalign*}
\end{proof}

Now, let us prove \Clm{cutvalbnd}.

\begin{proof}[Proof of \Clm{cutvalbnd}.] 
    We call a phase $h$ significant if $|F_h| \leq \sizefrac n$. By definition, the contribution to the cut-bound from the insignificant phases is at most $\sizefrac nd$. Now, we control the contribution to cut-bound from the significant phases. To this end, let $h_c$ denote the last significant phase. By \Prop{bwise:basic}, the probability that $h_c \leq 2/\delta$ is at least $1 - n^{-b/4}$.

    Now we do a phase-by-phase analysis. Fix any significant phase $h \leq h_c$. Condition on any choice of randomness in $\bS_1 \circ \bS_2$ that results in $|F_h| \geq \sizefrac n$. Note that all calls to $\cluster(v,t_v,k_v)$ issued in this phase return a set $C_v$ which satisfies

    $$|C_v| = \begin{cases}
        = 1 \ \text{ if } C_v = v\\
        \leq 2 k_h \ \text{Otherwise} 
    \end{cases}$$

    Furthermore, by an application of \Cor{kwise}, with probability at least $1 - n^{-b/4}$, at most $\frac{3 \delta n}{2} \leq 2 \delta n$ calls to $\cluster(v)$ are issued in all phases $h$ preceding $h_c$. Thus, the unclustered set when phase $h$ satisfies 
    
    \begin{equation} \label{eq:fhplus1}
    |F_{h+1}| = \sizefrac n - \sum_{v \in V_h} |C_v| \geq \sizefrac n - 2\delta n \cdot 2 k_h \geq \sizefrac n/2 \text{ vertices.}
    \end{equation}
    
    We denote by $F_v$ the set of free vertices which remain when the vertex $v$ is processed by $\globpart(\bS_1 \circ \bS_2)$. Toward controlling the contribution to the cut-bound in the $h${th} phase, we would like to control $\sum_{v \in V_h} |C_v \cap F_v|$. Recalling that $|V_h| \leq 2 \delta n$, we have for each $v \in V_h$, $$|E(C_v, \overline{C_v})| \leq 2 \phi k_h d + d = 4 \phi k_h d^2.$$

    Now we consider the contribution to the number of significant edges cut in phase $h$ from a vertex $v \in V_h$. This is seen to be

     \begin{align*} 
    e_h(v) \eqdef |E(C_v \cap F_v, \overline{C_v \cap F_v})| \ \ &\substack{\bone\\ \leq} \sum_{e = (x,y) \in E(C_v, \overline{C_v)}} (d(x) + d(y))  \\
    &\leq \sum_{e = (x,y) \in E(C_v, \overline{C_v)}} 2d  \\
    &\leq 2d \cdot |E(C_v, \overline{C_v})| \\
    &\leq 8 \phi k_h d^2
    \end{align*}

    Here, $\bone$ follows because for every edge $(x,y) \in E(C_v, \overline{C_v})$, we can have the entire set $E(v=x, N(x))$ can become a subset of the edges crossing from $C_v \cap F_v$ to its complement (and ditto for $E(y, N(y))$).

    Therefore, with probability at least $1 - 2 \cdot n^{-b/4}$, the number of significant edges cut in this phase is given as  

    \begin{equation} \label{eq:sig:edges:bound}
        e_h \eqdef \sum_{v \in V_h} e_h(v) \leq 8 \phi k_h d^2 \cdot |V_h| \leq 16 \phi k_h d^2 \delta n
    \end{equation}

    Next, we show another way to think about the quantity $\delta n k_h$. The key random variable to track the number of free vertices returned in each call to $\cluster(v)$ for $v \in V_h$. That is, we define $X_h \eqdef \sum_{v \in V_h} |\mathtt{cluster}(v) \cap F_v|$. By \Eqn{fhplus1}, after phase $h$ is processed, there are still at least $\sizefrac n/2$ vertices left. Denote by $\cE$ the event that at least $\sizefrac^5 \cdot\delta n/32$ vertices in $V_h$ are viable. By \Thm{findrv2}, the event $\cE$ occurs with probability at least $1 - \exp(1/\eps)$. Under the event $\cE$, we note

    \begin{align*}
        X_h &= \sum_{v \in V_h} |C_v \cap F_v| \geq \sum_{\substack{v \in V_h \\v \text{ viable } }}|C_v \cap F_v| \geq \sizefrac^5 \delta n/32 \cdot \sizefrac^3 k_h/8 &&\text{Definition of viable}\\
        &= \frac{\sizefrac^8 \delta k_h n}{256}
    \end{align*}

    \noindent Together with \Eqn{sig:edges:bound}, this gives $$e_h \leq 16 d^2 \cdot \frac{256 X_h}{\sizefrac^8} \leq 4096 \phi \sizefrac^{-8} d^2 X_h \eqdef \alpha_h$$ holds with probability at least $1 - 2n^{-b/4} - \exp(1/\eps)$. Let us call the event under which the number of significant edges cut in a significant phase $h \leq h_c$ is at most $\alpha_h$ by $\cF_h$. We can now write

    \begin{align*}
        \sum_h \EX[\textrm{\# significant edges cut}] &\leq 4096 \phi \sizefrac^{-8} d^2 \sum_h \EX[X_h | \cF_h] \cdot Pr(\cF_h) + \Pr(\overline{\cF_h}) \cdot nd \\
        &\leq 4096 \phi \sizefrac^{-8} d^2 \sum_h \EX[X_h] + (\overline{h} \cdot (\exp(1/\eps) + 2n^{-b/4}) \cdot nd \\
        &\leq 4096 \phi \sizefrac^{-8} d^2 \Bigg(\sum_{h \leq \overline{h}} \EX\Big[\sum_{v \in V_h} |\mathtt{cluster}(v) \cap F_h|\Big] \Bigg) + \sizefrac nd/2
    \end{align*}
\end{proof}

To finish up, we prove \Clm{newcharging}. 

\begin{proof}[Proof of \Clm{newcharging}.]
To prove this claim, we follow a charging argument similar to one in \cite{KSS:21}. 
We charge one unit to every vertex $v$ in $\texttt{cluster}(v) \cap F_h$ when $u$ is processed in $\texttt{globalPartition}(\bS_1 \circ \bS_2)$. 
This essentially implies a vertex $u$ receives charge in at most one phase, precisely when it leaves the free set. 
The total amount of charge is the quantity on the left-hand side of the inequality that we aim to bound. 
We will complete our argument by showing that the expected amount of charge received by any vertex is at most $4$ units. 

Let $X_u$ be the random variable that denotes the charge a fixed vertex $u$ receives. We compute the expectation of the random variable in the following way:
$$ \myex[X_u] = \sum_h \myex[X_u| u \text{ received charge in phase } h] \cdot \mypr[u \text{ received charge in phase } h]$$
Our aim is to show $\myex[X_u| u \text{ received charge in phase } h] \leq 4$, which implies the desired inequality. \\

Now, $u$ receives charge in phase $h$ when there exists a $v \in V_h$ such that $u \in \texttt{cluster}(v)$. 
The charge received by $u$ equals $c \eqdef |S(u)|$ where $S(u) = \{v \in V_h \colon \cluster(v) \ni u \}$. Observe that, $v \in \ball(u)$ and by \Clm{ballsize} $c = |S(u)| \leq \len\minp^{-1}$. Recall that the collection of phase $h$ Bernoullis, $\{B^h(w)\}_{w \in V}$ is $b$-wise independent with $b = \frac{4 \ell}{\rho} \geq 4c$. Thus, for any $v \in S(u)$, the events that $v \in V_h$ are all independent. Therefore, we can bound

$$\Pr_{\bS_1 \circ \bS_2}\Bigg[ u \text{ received charge in phase } h \Big\vert u \in F_h \Bigg] = 1 - (1-\delta)^c .$$

Moreover, we also note

\begin{align*}&  \ \ \ \delta c \leq \delta \len\minp^{-1} = (d^{-4}\cdot \varepsilon^{100}) \leq 1/2 \qquad \text{and} \\
&\therefore (1-\delta)^c \leq 1-\delta c + (\delta c)^2 \leq 1 - \delta c/2
\end{align*}
The above two together imply, 
$$\Pr_{\bS_1 \circ \bS_2}\Bigg[ u \text{ received charge in phase } h \Big\vert u \in F_h \Bigg] \geq \delta c/2$$ 
Note that
$$\EX\Bigg[X_u \Big\vert u \in F_h \Bigg] = \sum_{b > 0} {c \choose b} \delta^b \leq \sum_{b > 0} (\delta c)^b \leq 2\delta c$$
Now observe 
\begin{align*}
&\bone. \quad \EX\Bigg[\Big(X_u \big\vert u \text{ received charge in phase } h\Big) \Bigg\vert u \in F_h \Bigg] \leq (2\delta c)/(\delta c/2) \leq 4 \\
&\btwo. \quad \EX\Bigg[X_u \big\vert u \notin F_h \Bigg] = 0
\end{align*}

\noindent This means

$$\myex[X_u| u \text{ received charge in phase } h]\leq 4$$
\end{proof}

\section{Proof of \texorpdfstring{\Thm{main:huge-label}}
{Theorem~\ref{thm:main:huge-label}}}

\begin{proof}
    By construction, all the randomness of the $\globpart(\bS_1 \circ \bS_2)$ procedure lies in $\bS_1$ and $\bS_2$. As noted in \Rem{s1-randomness}, the total amount of randomness consumed by $\bS_1$ is $O(\overline{h} b \cdot \log \ell \cdot \log N) = O\Big(\frac{\log N}{\poly(\eps)}\Big)$.

    And as one notes from the code of $\findr$, the number of random bits used in $\bS_2$ is also $O\Big(\frac{\log n}{\poly(\eps)}\Big)$. This proves that the total number of bits used by $\globpart$ is of the same order. Moreover, by \Thm{findrv2}, \Thm{edgecut:bnd}, and \Rem{remainingtasks} $\globpart(\bS_1 \circ \bS_2)$ simulates $\globpart(\bR)$ algorithm from \cite{KSS:21} with probability at least $1 - \exp(1/\eps)$ which means $\globpart(\bS_1 \circ \bS_2)$ implements a bonafide $\poly(\eps^{-1})$ time partition oracle. 
\end{proof}

\begin{remark}
Note that \Thm{main} follows as an immediate corollary of the preceding result by setting $N = n$.
\end{remark}

\section{Acknowledgement}
Part of this project began when Abhiruk Lahiri was at Charles University in Prague, Czechia. He would like to thank Zdenek Dvorak for several discussions on minor-free graphs that continued beyond his stay in Prague. 

\bibliographystyle{alpha}

\end{document}